\documentclass[submission,copyright,creativecommons]{eptcs}
\providecommand{\event}{QPL 2025}
\pdfoutput=1

\usepackage{iftex}
\usepackage{braket}
\usepackage{amsthm, amsmath, amssymb, dsfont, physics, textcomp, tikz}
\usepackage{quantikz}
\usepackage{underscore}
\usepackage{nicematrix}
\usepackage{xurl}
\usepackage{hyperref}
\usepackage[capitalize,nameinlink,sort,noabbrev]{cleveref}

\newcommand{\z}{\zeta}
\newcommand{\Z}{\mathbb{Z}}
\newcommand{\N}{\mathbb{N}}
\newcommand{\C}{\mathbb{C}}
\newcommand{\D}{\mathbb{D}}
\newcommand{\zt}{\zeta_{12}}
\newcommand{\Zzt}{\mathbb{Z}[\zeta_{12}]}
\newcommand{\Dzt}{\mathcal{R}_{12}}
\newcommand{\dzt}{\delta}

\newcommand{\Q}{\mathbb{Q}}
\DeclareMathOperator{\diag}{diag}
\DeclareMathOperator{\sgn}{sgn}
\DeclareMathOperator{\lde}{lde}
\newcommand{\s}[1]{\{#1\}}
\newcommand{\matrices}{\mathrm{M}}
\newcommand{\unitaries}{\mathrm{U}}

\newtheoremstyle{break}% name
  {}%          Space above, empty = `usual value'
  {}%          Space below
  {\itshape}%  Body font
  {}%          Indent amount (empty = no indent, \parindent = para indent)
  {\bfseries}% Thm head font
  {.}%         Punctuation after thm head
  {\newline}%  Space after thm head: \newline = linebreak
  {}%          Thm head spec
\theoremstyle{plain}
\newtheorem{theorem}{Theorem}[section]
\newtheorem{lemma}[theorem]{Lemma}
\newtheorem*{lemma*}{Lemma}
\newtheorem{proposition}[theorem]{Proposition}
\newtheorem*{proposition*}{Proposition}
\newtheorem{corollary}[theorem]{Corollary}
\theoremstyle{break}

\theoremstyle{definition}

\theoremstyle{remark}
\newtheorem{remark}[theorem]{Remark}

\ifpdf
  \usepackage{underscore}        % Only needed if you use pdflatex.
  \usepackage[T1]{fontenc}       % Recommended with pdflatex
\else
  \usepackage{breakurl}          % Not needed if you use pdflatex only.
\fi

\newcommand{\urlaltve}[2]{\href{#2}{\nolinkurl{#1}}}

\usepackage{xcolor}
\hypersetup{
    colorlinks,
    linkcolor={red!80!black},
    citecolor={green!80!black},
    urlcolor={blue!80!black}
}

\title{
  Contributions to the Theory of Clifford-Cyclotomic Circuits
  }

\author{Linh Dinh \& Neil J. Ross
\institute{Dalhousie University}}
\def\titlerunning{Contributions to the Theory of Clifford-Cyclotomic Circuits}
\def\authorrunning{Linh Dinh \& Neil J. Ross}

\begin{document}
\maketitle

% --------------------------------------------------------------------
\begin{abstract}
  Let $n$ be a positive integer divisible by 8. The
  Clifford-cyclotomic gate set $\mathcal{G}_n$ consists of the
  Clifford gates, together with a $z$-rotation of order $n$. It is
  easy to show that, if a circuit over $\mathcal{G}_n$ represents a
  unitary matrix $U$, then the entries of $U$ must lie in
  $\mathcal{R}_n$, the smallest subring of $\mathbb{C}$ containing
  $1/2$ and $\mathrm{exp}(2\pi i/n)$. The converse implication, that
  every unitary $U$ with entries in $\mathcal{R}_n$ can be represented
  by a circuit over $\mathcal{G}_n$, is harder to show, but it was
  recently proved to be true when $n=2^k$. In that case, $k-2$
  ancillas suffice to synthesize a circuit for $U$, which is known to
  be minimal for $k=3$, but not for larger values of $k$. In the
  present paper, we make two contributions to the theory of
  Clifford-cyclotomic circuits. Firstly, we improve the existing
  synthesis algorithm by showing that, when $n=2^k$ and $k\geq 4$,
  only $k-3$ ancillas are needed to synthesize a circuit for $U$,
  which is minimal for $k=4$.  Secondly, we extend the existing
  synthesis algorithm to the case of $n=3\cdot 2^k$ with $k\geq 3$.
\end{abstract}

% --------------------------------------------------------------------
\section{Introduction}
\label{sec:intro}

% --------------------------------------------------------------------
\subsection{Background}
\label{ssec:background}

Let $n$ be a positive integer divisible by 8. The
\textbf{Clifford-cyclotomic gate set of degree $n$}, which we denote
by $\mathcal{G}_n$, consists of the usual \textbf{Clifford gates}
    \[
    H = \frac{1}{\sqrt{2}}\begin{bmatrix}
        1 & 1 \\
        1 & -1
    \end{bmatrix}, 
    \qquad 
    S = \begin{bmatrix}
        1 & 0 \\
        0 & i
    \end{bmatrix}, 
    \qquad
    \mbox{and}
    \qquad
    CX = \begin{bmatrix}
        1 & 0 & 0 & 0 \\
        0 & 1 & 0 & 0 \\
        0 & 0 & 0 & 1 \\
        0 & 0 & 1 & 0
    \end{bmatrix}, 
    \]
together with the $z$-rotation of order $n$   
    \[
    T_n = \begin{bmatrix}
        1 & 0 \\
        0 & \zeta_n
    \end{bmatrix},
    \]
where $\z_{n}$ is the \textbf{primitive $n$-th root of unity} $\z_{n}=
e^{2\pi i /n}$. For any $n$ divisible by 8, $\mathcal{G}_n$ is
universal for quantum computation. The gate sets $\mathcal{G}_8$ and
$\mathcal{G}_{16}$ are also known as the \textbf{Clifford+$T$} and
\textbf{Clifford+$\sqrt{T}$} gate sets,
respectively. Clifford-cyclotomic circuits are important to the design
of quantum algorithms \cite{shor}, the theory of fault-tolerant
quantum computation \cite{distillation,hierarchy}, and the study of
quantum complexity \cite{qma}. Because of this, they have received
significant attention in the literature
\cite{qubitcyclo,Amy2020,forest2015exact,Giles2013a,qutritcyclo,ingallsjordankeetonloganzaytman2021,kliuchnikov2013fast}.

An important question in the theory of Clifford-cyclotomic circuits is
that of precisely characterizing the matrices that can be exactly
represented by a circuit over $\mathcal{G}_n$. Let $\mathcal{R}_n$ be
the smallest subring of $\mathbb{C}$ containing $1/2$ and $\z_{n}$. It
is easy to see that, if a circuit over $\mathcal{G}_n$ represents a
unitary matrix $U$, then the entries of $U$ must lie in
$\mathcal{R}_n$. The converse implication, that every unitary $U$ with
entries in $\mathcal{R}_n$ can be represented by a circuit over
$\mathcal{G}_n$, is harder to show. Indeed, until recently, it was
only known to be true for the Clifford+$T$ gate set $\mathcal{G}_8$,
as well as for a handful of adjacent gate sets
\cite{Amy2020,Giles2013a}. Recently, however, it was shown that this
converse is true whenever $n$ is a power of 2 \cite{qubitcyclo}. In
that case, $\log(n)-2$ ancillas suffice to synthesize a circuit for
$U$, which is known to be minimal for $n=8$, but not for larger powers
of 2. This recent exact synthesis result naturally raises two
questions. Firstly, when $n$ is a power of 2, can the number of
ancillas needed to synthesize a Clifford-cyclotomic circuit of degree
$n$ be made smaller than $\log(n)-2$? Secondly, can one prove an exact
synthesis result for Clifford-cyclotomic circuits whose degree is not
a power of 2? In the present paper, we contribute to the theory of
Clifford-cyclotomic circuits by answering both questions positively.

% --------------------------------------------------------------------
\subsection{Contributions}
\label{ssec:contribs}

Let $m$ be a positive integer.

We prove that any $2^m$-dimensional unitary $U$ with entries in
$\mathcal{R}_{16}$ can be exactly represented by an $m$-qubit circuit
over $\mathcal{G}_{16}$ using at most $1$ ancilla, which is
minimal. We establish this result through a study of the effect of
catalytic embeddings \cite{catemb} on the determinant which may be of
independent interest. We then use this result about circuits over
$\mathcal{G}_{16}$ to save an ancilla in synthesizing circuits over
$\mathcal{G}_{2^k}$, for $k\geq 4$.

We also prove, inspired by the results of \cite{twelfth}, that any
$2^m$-dimensional unitary $U$ with entries in $\mathcal{R}_{3\cdot
  2^k}$ with $k \geq 3$ can be exactly represented by an $m$-qubit
circuit over $\mathcal{G}_{3\cdot 2^k}$, thereby establishing a
number-theoretic characterization for Clifford-cyclotomic circuits
whose degree is not a power of 2.

% --------------------------------------------------------------------
\section{Rings}
\label{sec:nt}

We now introduce the rings that will be important in what follows and
we discuss some of their properties. For further details, we encourage
the reader to consult \cite{dandf,cyclo}.

We assume that rings have a multiplicative identity. If $R$ is a ring
and $u\in R$, we write $R/(u)$ for the \textbf{quotient of $R$ by the
  ideal $(u)$}. Two elements $v$ and $v'$ of the ring $R$ are
\textbf{congruent modulo $u$}, if their difference is a multiple of
$u$, that is, if $v-v'\in (u)$. In that case, we write $v\equiv_u v'$,
or $v\equiv v' \pmod{u}$. The relation $\equiv_u$ is an equivalence
relation on $R$ and the elements of $R/(u)$ are precisely the
equivalence classes of elements of $R$ under the relation
$\equiv_u$. We sometimes refer to these equivalence classes as
\textbf{residues}.

% --------------------------------------------------------------------
\subsection{Cyclotomic integers}
\label{ssec:cyclo}

We write $\zeta_n$ for the \textbf{primitive $n$-th root of unity}
$\z_{n} = e^{2 \pi i /n}$. The \textbf{ring of cyclotomic integers
  $\Z[\z_{n}]$} is the smallest subring of $\C$ that contains
$\z_{n}$. Since $\zeta_n^\dagger = \zeta_n^{n-1}$, the ring $\Z[\z_n]$
is closed under complex conjugation.

Let $\varphi$ denote \textbf{Euler's totient function}, so that
$\varphi(n)$ counts the integers in $\s{1,\ldots, n}$ that are
relatively prime to $n$. The ring $\Z[\z_{n}]$ can be characterized as
\begin{equation}
  \label{eq:cycloring}
  \Z[\z_n] = \left\{ a_0 + a_1\z_n + a_2\z_n^2 + \cdots +
  a_{\varphi(n)-1}\z_n^{\varphi(n)-1} \mid a_0,\ldots,
  a_{\varphi(n)-1} \in \Z \right\}.
\end{equation}
Every element $u\in\Z[\z_n]$ can be uniquely expressed as a
$\Z$-linear combination of powers of $\z_n$ as in
\cref{eq:cycloring}. That is, we have
\[
a_0 + a_1\z_n +  \cdots +
  a_{\varphi(n)-1}\z_n^{\varphi(n)-1} = a_0' + a_1'\z_n +  \cdots +
  a_{\varphi(n)-1}'\z_n^{\varphi(n)-1}
\]
if and only if $a_j=a_j'$ for $0\leq j \leq \varphi(n)-1$.

We will be interested in two families of rings of cyclotomic integers:
the one corresponding to $n=2^k$, and the one corresponding to
$n=3\cdot 2^k$. For $k\leq k'$, we have $\Z[\zeta_{2^k}]\subseteq
\Z[\zeta_{2^{k'}}]$ (resp. $\Z[\zeta_{3\cdot 2^k}]\subseteq
\Z[\zeta_{3\cdot 2^{k'}}]$). Moreover, for $k\geq 1$, every element
$u\in\Z[\z_{2^{k+1}}]$ (resp. $u\in \Z[\z_{3\cdot 2^{k+1}}]$) can be
uniquely written as $u=a+b\z_{2^{k+1}}$ (resp. $u=a+b\z_{3\cdot
  2^{k+1}}$), with $a,b\in \Z[\z_{2^{k}}]$ (resp. $a,b\in
\Z[\z_{3\cdot 2^{k}}]$).

We define the ring $\mathcal{R}_n$ as the smallest subring of $\C$
that contains $1/2$ and $\zeta_n$. The ring $\mathcal{R}_n$ can be
characterized as in \cref{eq:cycloring}. Indeed, we have
\begin{equation}
  \label{eq:cycloring2}
  \mathcal{R}_n = \left\{ a_0 + a_1\z_n + a_2\z_n^2 + \cdots +
  a_{\varphi(n)-1}\z_n^{\varphi(n)-1} \mid a_0,\ldots,
  a_{\varphi(n)-1} \in \D \right\},
\end{equation}
where $\D = \s{a/2^\ell \mid a\in \Z \mbox{ and } \ell\in \N}$ is the
ring of \textbf{dyadic fractions}. Every element $u\in\mathcal{R}_n$
can be uniquely expressed as in \cref{eq:cycloring2}.

If $n$ is divisible by 8, then $n=8d$, so that $\z_n^{2d}=\z_4=i$,
$\z_{n}^d = \z_8$, and $\z_n^d + \z_n^{-d} = \z_8 + \z_8^\dagger
=\sqrt{2}$. Hence, in that case, we have $i\in\mathcal{R}_n$ and
$1/\sqrt{2}\in \mathcal{R}_n$, which implies that the entries of the
gates $H$, $S$, $CX$, and $T_n$ belong to $\mathcal{R}_n$. Thus,
whenever $n$ is divisible by 8, any matrix that can be represented by
a circuit over $\mathcal{G}_n$ has entries in $\mathcal{R}_n$.

% --------------------------------------------------------------------
\subsection{The ring \texorpdfstring{$\Z[\z_{12}]$}{Z[zeta12]}}
\label{ssec:ztwelve}

The ring $\Z[\z_{12}]$ will play an important role in
\cref{sec:12}. We now record some of its relevant properties. We know
from \cref{eq:cycloring} that
\[
\Z[\z_{12}] = \left\{ a_0 + a_1\z_{12} + a_2\z_{12}^2 +a_3\z_{12}^3
\mid a_0, a_1, a_2, a_3 \in\Z \right\}.
\]

We define $\dzt \in \Z[\zt]$ as $\dzt = 1+ \zt^3 = 1+i$. The
cyclotomic integer $\dzt$ is prime in $\Z[\zt]$ and the prime
factorization of $2$ in $\Zzt$ is given by
\begin{equation}
\label{eq:factortwo}
2 = (1+i)^2(-i) = \dzt^2(-i).
\end{equation}
Now consider an element $u\in\mathcal{R}_{12}$. By
\cref{eq:cycloring2,eq:factortwo}, we can write $u$ as
$u=u'/\delta^\ell$, with $u'\in\Zzt$ and $\ell\in\N$. The smallest
such $\ell$ is called the \textbf{least denominator exponent of $u$}
and is denoted $\lde(u)$. Equivalently, $\lde(u)$ is the smallest
$\ell\in\N$ such that $\delta^\ell u \in \Zzt$. More generally, if $M$
is a matrix (or a vector) with entries in $\mathcal{R}_{12}$, then
$\lde(M)$ is the smallest $\ell$ such that $\delta^\ell M$ is a matrix 
(or a vector) over $\Zzt$.

\begin{proposition}
    \label{prop:12thquotient}
    We have:
    \begin{itemize}
        \item $\Zzt/(2) = \{a_0+a_1\zt+a_2\zt^2+a_3\zt^3 \mid a_0,
          a_1, a_2, a_3 \in \{0,1\}\}$, and
        \item $\Zzt/(\dzt) = \{0, 1, \zt, \zt^2\}$.
    \end{itemize}
\end{proposition}

\begin{proof}
  Let $u = a_0+a_1\zt+a_2\zt^2+a_3\zt^3 \in \Zzt$. Since $2\equiv
  0\pmod{2}$, the coefficients $a_0$, $a_1$, $a_2$, and $a_3$ can be
  chosen in $\s{0,1}$, which establishes the first item in the
  proposition. For the second item, notice that, since $\dzt =
  1+\zt^3$, we have $1+\zt^3\equiv 0\pmod{\dzt}$ so that $\zt^3 \equiv
  -1 \equiv 1\pmod{\dzt}$. This, together with the fact that $2\equiv
  0 \pmod{\dzt}$, implies that any $u\in\Zzt / (\dzt)$ can be written
  as $u=a_0+ a_1\zt+a_2\zt^2$ with $a_0,a_1,a_2\in\s{0,1}$. Since we
  have the cyclotomic relation $\zt^4 -\zt^2 + 1 = 0$, we get $\zt^2
  \equiv \zt + 1 \pmod{\dzt}$, from which the second item in the
  proposition then follows.
\end{proof}

The \textbf{quadratic integer ring $\Z[\sqrt{3}]$} is the smallest
subring of $\C$ containing $\Z$ and $\sqrt{3}$. The elements of
$\Z[\sqrt{3}]$ can be characterized as
$\Z[\sqrt{3}]=\s{a+b\sqrt{3}\mid a,b\in\Z}$ and every element of
$\Z[\sqrt{3}]$ can be uniquely written as such a $\Z$-linear
combination of $1$ and $\sqrt{3}$. Because $\sqrt{3} = \z_{12}
+\z_{12}^\dagger$, we have $\Z[\sqrt{3}]\subseteq \Z[\z_{12}]$. Since
$\z_{12}^6 = -1$ and $\z_{12}^4 - \z_{12}^2 + 1=0$, we have
$\z_{12}^\dagger = \zt - \zt^3$. It can then be verified by direct
computation that if $u = a+b\zeta_{12}+c\z_{12}^2+d\z_{12}^3 \in
\Z[\zeta_{12}]$, then
\begin{equation}
\label{eq:norm12}
u^\dagger u = ((a^2 + c^2 + ac) + (b^2 + d^2 + bd)) +
(ab + bc + cd)\sqrt{3}.
\end{equation}
In particular, the Euclidean norm of $u\in \Z[\zeta_{12}]$ belongs
to $\Z[\sqrt{3}]$, i.e., if $u\in \Z[\zeta_{12}]$, then $u^\dagger u
\in \Z[\sqrt{3}]$.

\begin{lemma}
\label{lem:res}
We have:
\begin{itemize}
    \item if $u\in\Zzt$ and $u \equiv_{\dzt} 0$, then $u^\dagger u
      \equiv_2 0$,
    \item if $u\in\Zzt$ and $u\not\equiv_{\dzt} 0 $, then $u^\dagger u
      \equiv_2 1 $ or $u^\dagger u \equiv_2 \sqrt{3}$, and
    \item if $u,v\in\Zzt$ and $u^\dagger u\equiv_2 v^\dagger v $, then
      $u \equiv_2 \zt^m v$ for some integer $m$.
\end{itemize}
\end{lemma}

\begin{table}
\centering
        \begin{tabular}{c|c|c}
            $u \pmod{\dzt}$ & $u \pmod{2}$ & $u^\dagger u \pmod{2}$
            \\
            \hline
            $0$ & $0$ & $0$ \\
            $0$ & $1+\zt^3$ & $0$
            \\
            $0$ & $1+\zt+\zt^2$ & $0$ \\
            $0$ & $\zt+\zt^2+\zt^3$ & $0$ \\
            \hline
            $1$ & $1$ & $1$ \\
            $1$ & $\zt^3$ & $1$ \\
            $1$ & $\zt+\zt^2 \equiv_2\zt(1+\zt)$ & $\sqrt{3}$ \\
            $1$ & $1+\zt+\zt^2+\zt^3\equiv_2\zt^4(1+\zt)$ & $\sqrt{3}$ \\
            \hline
            $\zt$ & $\zt$ & $1$ \\
            $\zt$ & $1+\zt^2\equiv_2\zt^4$ & $1$ \\
            $\zt$ & $\zt^2+\zt^3 \equiv_2  \zt^2(1+\zt)$ & $\sqrt{3}$ \\
            $\zt$ & $1+\zt+\zt^3 \equiv_2  \zt^5(1+\zt)$ & $\sqrt{3}$ \\
            \hline
            $\zt^2$ & $\zt^2$ & $1$ \\
            $\zt^2$ & $1+\zt$ & $
            \sqrt{3}$ \\
            $\zt^2$ & $1+\zt^2+\zt^3 \equiv_2 \zt^3(1+\zt)$ & $\sqrt{3}$ \\
            $\zt^2$ & $\zt+\zt^3 \equiv_2 \zt^5$ & $1$
        \end{tabular}
\caption{The possible residues for $u\in\Z[\z_{12}]$ modulo $\delta$
  and $2$, as well as those for $u^\dagger u$ modulo
  2. \label{tab:residues}}
\end{table}    

\begin{proof}
    By inspection of \cref{tab:residues}, which lists the possible
    residues of $u\in\Zzt$ modulo $\dzt$ and $2$, as well as the
    possible residues of $u^\dagger u$ modulo 2. The calculations
    leading to the construction of \cref{tab:residues} can be verified
    using the congruences from the proof of \cref{prop:12thquotient},
    in addition to the following facts: $\dzt^\dagger\dzt = 2$,
    $(1+\zt)^\dagger(1+\zt) = 2 + \sqrt{3} \equiv_2 \sqrt{3}$, and
    $\zt^4 \equiv_2 \zt^2+1$.
\end{proof}

% --------------------------------------------------------------------
\section{Matrices and circuits}
\label{sec:mats}

Let $R$ be a commutative ring. We write $\matrices(R)$ for the
collection of all square matrices with entries in $R$, and
$\matrices_m(R)$ for the ring of $m\times m$ matrices in $R$. When $R$
is a subring of $\mathbb{C}$ that is closed under complex conjugation,
we write $\unitaries(R)$ for the collection of all unitary matrices
with entries in $R$, and $\unitaries_m(R)$ for the group of $m\times
m$ unitary matrices with entries in $R$.

We now introduce certain matrices which will be useful in
\cref{sec:12}. Let $c\in \C$. For $0\leq j \leq m-1$, the
\textbf{one-level operator of type $c$} is the $m\times m$ matrix
$c_{[j]}$ defined as
\[
c_{[j]} = 
\begin{bNiceMatrix}[first-row,first-col]
          & \cdots & j  & \cdots  \\
\vdots & I         & 0 & 0   \\
j          & 0        & c & 0   \\
\vdots & 0        & 0 & I \\
\end{bNiceMatrix}
\]
Similarly, let $M\in\matrices_2(\C)$. For $0\leq j < j' \leq m-1$, the
\textbf{two-level operator of type $M$} is the $m\times m$ matrix
$M_{[j,j']}$ defined as
\[
M_{[j,j']} = 
\begin{bNiceMatrix}[first-row,first-col]
       & \cdots & j      & \cdots & j'      & \cdots \\
\vdots & I      & 0      & 0      & 0       & 0        \\
j      & 0      & M_{1,1} & 0      &  M_{1,2} & 0        \\
\vdots & 0      & 0      & I      &  0      & 0        \\
j'     & 0      & M_{2,1} & 0      &  M_{2,2} & 0        \\
\vdots & 0      & 0      & 0      &   0     & I         \\
\end{bNiceMatrix}
\]
The one-level operator $c_{[j]}$ acts on an $m$-dimensional vector
$\textbf{u}$ by scaling its $j$-th component by $c$ and leaving the
remaining components unchanged. The two-level operator $M_{[j,j']}$
similarly acts as $M$ on the $j$-th and $j'$-th components of
$\textbf{u}$ and leaves the remaining components unchanged. Note that
if $|c|=1$ then $c_{[j]}$ is unitary, and that if
$M\in\unitaries_2(\C)$, then $M_{[j,j']}$ is unitary. 

We close this section by recalling the existence of some well-known
circuit constructions which will be useful in what follows.  As
mentioned in \cref{sec:intro}, when $n$ is divisible by 8, the gate
set $\mathcal{G}_n$ subsumes the Clifford+$T$ gate set. Hence, any
matrix that can be represented by a Clifford+$T$ circuit can also be
represented by a circuit over $\mathcal{G}_n$. As a consequence, the
one-level operators of type $\z_{n}$ and the two-level operators of
type $X$ and $H$ can be represented exactly over the gate set
$\mathcal{G}_n$ using a single ancilla.

\begin{theorem}[Giles \& Selinger]
\label{thm:gstwo}
The one- and two-level operators of type $\z_n$, $X$, and $H$ can each
be exactly represented by a circuit over $\mathcal{G}_n$ using at most
one ancilla.
\end{theorem}

\begin{proof}
Circuit constructions for the one- and two-level operators of type
$\z_8$, $X$, and $H$ using a single ancilla can be found in
\cite[Section~5]{Giles2013a}. To obtain a circuit for the one-level
operator of type $\z_n$, it suffices to replace $T$ with $T_n$ where
appropriate.
\end{proof}

% --------------------------------------------------------------------
\section{Determinants}
\label{sec:dets}

The \textbf{determinant} is an important function on matrices. We will
be interested in the determinant of certain block matrices.

Let $R$ be a commutative ring and let $M$ be an $m\times m$ matrix
with entries in $R$. The \textbf{determinant of $M$ with respect to
  $R$}, denoted by $\det\nolimits_R(M)$, is defined as
\[
\det\nolimits_R(M) = \sum_{\sigma\in S_m} \sgn(\sigma) \prod_{j=1}^m
M_{j,\sigma(j)},
\]
where $S_m$ is the symmetric group of degree $m$ and $\sgn(\sigma)$ is
the sign of the permutation $\sigma\in S_m$.  If $R$ is a subring of
some ring $R'$ and $M$ is a matrix with entries in $R$, we have
$\det\nolimits_R(M) = \det\nolimits_{R'}(M)$. For simplicity, when the
ring $R$ in which the determinant is to be computed is clear from
context, we sometimes write $\det(M)$ rather than
$\det\nolimits_R(M)$. The determinant is multiplicative, in the sense
that for matrices $M$ and $M'$ over $R$ of compatible dimensions, we
have $\det\nolimits_R(MM') =
\det\nolimits_R(M)\det\nolimits_R(M')$. Moreover, if $M$ is a complex
unitary matrix, then $|\det\nolimits_\C(M)|=1$.

If $\phi:R\to R'$ is a ring homomorphism, then the entrywise
application of $\phi$ is a ring homomorphism $\matrices_m(R) \to
\matrices_m(R')$. By a slight abuse of notation, we use the symbol
$\phi$ to denote both the homomorphism $R\to R'$ and its entrywise
extension $\matrices(R) \to \matrices(R')$. Now consider a matrix $M$
over $R$. Since $\phi$ is a ring homomorphism and the determinant is a
polynomial in the matrix entries, we have
\begin{equation}
\label{eq:dethom}
\phi(\det\nolimits_R(M)) 
= \phi \left(\sum_{\sigma \in S_m} \sgn(\sigma) \prod_{j=1}^m M_{j,\sigma(j)}\right) 
= \sum_{\sigma \in S_m} \sgn(\sigma) \prod_{j=1}^m \phi\left( M\right)_{j,\sigma(j)}  
= \det\nolimits_{R'} (\phi(M)).
\end{equation}

Let $R$ be a commutative ring and let $R'$ be a commutative subring of
$\matrices_m(R)$. Now consider a matrix $M$ in
$\matrices_{m'}(R')$. The matrix $M$ is an $m'\times m'$ block matrix
whose blocks are $m\times m$ matrices over $R$. We can therefore
compute the determinant of $M$ with respect to $R'$, and then the
determinant of the resulting matrix with respect to
$R$. Alternatively, we can ``open'' the blocks and think of $M$ as an
$mm'\times mm'$ matrix over $R$ to directly compute the determinant of
$M$ with respect to $R$. The following theorem, whose proof can be
found in \cite[Theorem~1]{Silvester2000}, states that these two
quantities are equal.

\begin{theorem}[Silvester]
\label{thm:detblock}
    Let $R$ be a commutative ring, let $R'$ be a commutative subring
    of $\matrices_m(R)$, and let $M \in \matrices_{m'}(R')$. Then
    \[
    \det\nolimits_R M = \det\nolimits_R(\det\nolimits_{R'} M).
    \]
\end{theorem}

It will be convenient for us to consider a variant of
\cref{thm:detblock} where $R'$ is not a subring of $\matrices_m(R)$
but simply related to one.

\begin{corollary}
\label{cor:detrel}
Let $R$ and $R'$ be two commutative rings, let $\phi : R' \to
\matrices_m(R)$ be a ring homomorphism, and let $M\in
\matrices_{m'}(R')$. Then
\[
\det\nolimits_R(\phi(M)) = \det\nolimits_R(\phi(\det\nolimits_{R'}(M))).
\]
\end{corollary}

\begin{proof}
    Let $Q = \phi[R']$ be the direct image of $R'$ under $\phi$. Then
    $Q$ is a commutative subring of $\matrices_m(R)$, since $R'$ is
    commutative. Hence, by \cref{thm:detblock} and \cref{eq:dethom},
    \[
    \det\nolimits_R(\phi(M)) = \det\nolimits_R(\det\nolimits_{Q}(\phi(M)))
    = \det\nolimits_R (\phi(\det\nolimits_{R'}(M))),
    \]
    as desired.
\end{proof}

% --------------------------------------------------------------------
\section{Catalytic embeddings}
\label{sec:catemb}

We now introduce \textbf{catalytic embeddings}
\cite{catemb,qubitcyclo,qutritcyclo}. Let $\mathcal{U}$ and
$\mathcal{V}$ be two sets of unitary matrices. An
$m$-\textbf{dimensional catalytic embedding} from $\mathcal{U}$ into
$\mathcal{V}$ is a pair $(\phi, \textbf{c})$ where $\phi: \mathcal{U}
\rightarrow \mathcal{V}$ is a function and $\textbf{c} \in \C^m$ is a
unit vector such that
    \begin{enumerate}
        \item If $U \in \mathcal{U}$ has dimension $d$, then $\phi(U)
          \in \mathcal{V}$ has dimension $md$, and
        \item For any $\textbf{u} \in \C^{d}$,
          $\phi(U)(\textbf{u}\otimes\textbf{c}) =
          (U\textbf{u})\otimes\textbf{c}.$
    \end{enumerate}
We often write $(\phi, \textbf{c}): \mathcal{U} \rightarrow
\mathcal{V}$ to indicate that $(\phi,\textbf{c})$ is a catalytic
embedding from $\mathcal{U}$ to $\mathcal{V}$. The composition of an
$m$-dimensional catalytic embedding $(\phi ,
\textbf{c}):\mathcal{U}\to\mathcal{V}$ and an $m'$-dimensional
catalytic embedding $(\phi', \textbf{c}'):\mathcal{V}\to\mathcal{W}$
is the $mm'$-dimensional catalytic embedding $(\phi'\circ \phi,
\textbf{c}\otimes \textbf{c}'):\mathcal{U}\to\mathcal{W}$. The
catalytic embedding $(\mathrm{id}_\mathcal{U},
1):\mathcal{U}\to\mathcal{U}$ acts as the identity for this notion of
composition.

We now discuss two specific catalytic embeddings which will be
important in the rest of this paper. Let $k\geq 2$ and define the
matrix $\Lambda_k$ and the vector $\textbf{c}_k$ by
    \[
    \Lambda_{k} = \begin{bmatrix}
        0 & 1 \\
        \zeta_{2^{k-1}} & 0
    \end{bmatrix} 
    \qquad
    \mbox{and}
    \qquad
    \textbf{c}_k = \frac{1}{\sqrt{2}}\begin{bmatrix}
        1\\
        \zeta_{2^k}
    \end{bmatrix}.
    \]
We can use $\Lambda_k$ and $\textbf{c}_k$ to define a catalytic
embedding, following \cite{catemb,qubitcyclo}. Consider a matrix
$M\in\matrices(\mathcal{R}_{2^k})$. Then $M$ can be uniquely written
as $M=A+B\z_{2^k}$, for some
$A,B\in\matrices(\mathcal{R}_{2^{k-1}})$, so that we can construct
the matrix $A\otimes I_2 + B\otimes
\Lambda_k\in\matrices(\mathcal{R}_{2^{k-1}})$. It can be shown that
the assignment
\begin{equation}
\label{eq:phik}
A+B\zeta_{2^k} \longmapsto A \otimes I_2 + B \otimes \Lambda_{k}
\end{equation}
defines, for every $m$, a ring homomorphism
$\phi_k:\matrices_m(\mathcal{R}_{2^{k}}) \to
\matrices_{2m}(\mathcal{R}_{2^{k-1}})$. Moreover, we have
$\phi_k(M^\dagger) = \phi_k(M)^\dagger$ for every $M$, so that
$\phi_k$ restricts to a group homomorphism
$\unitaries_m(\mathcal{R}_{2^k})\to\unitaries_{2m}(\mathcal{R}_{2^{k-1}})$. Now
let $\textbf{u}$ be an arbitrary vector. Then
\begin{align*}
\phi_k(U) (\textbf{u}\otimes\textbf{c}_k) &= (A\otimes I + B \otimes \Lambda_{k}) (\textbf{u}\otimes\textbf{c}_k)\\
&= (A\otimes I) (\textbf{u}\otimes\textbf{c}_k)+ (B \otimes \Lambda_{k})(\textbf{u}\otimes\textbf{c}_k)\\
&= A\textbf{u}\otimes I \textbf{c}_k+ B\textbf{u} \otimes \Lambda_{k}\textbf{c}_k \\
&= A\textbf{u}\otimes \textbf{c}_k+ B\textbf{u} \otimes \z_{2^k}\textbf{c}_k \\ 
&= A\textbf{u}\otimes \textbf{c}_k+ B\z_{2^k} \textbf{u} \otimes \textbf{c}_k \\ 
&= (A\textbf{u}+ B\z_{2^k} \textbf{u})\otimes\textbf{c}_k \\ 
&=(U\textbf{u})\otimes \textbf{c}_k. 
\end{align*}
The pair $(\phi_k, \textbf{c}_k)$ is therefore a catalytic
embedding. For future reference, we record this fact in the
proposition below.

\begin{proposition}
    \label{prop:cyclocatemb1}
    Let $k \geq 2$. Then $(\phi_{k}, \textbf{c}_k)$ is a
    $2$-dimensional catalytic embedding from
    $\unitaries(\mathcal{R}_{2^k})$ to
    $\unitaries(\mathcal{R}_{2^{k-1}})$.
\end{proposition}

By identifying $\matrices_1(\mathcal{R}_{2^k})$ with
$\mathcal{R}_{2^k}$, we can think of the function $\phi_k$, when
restricted to $\mathcal{R}_{2^k}$ as a ring homomorphism
$\mathcal{R}_{2^k} \to \matrices_2(\mathcal{R}_{2^{k-1}})$. The
function $\phi_k$ as defined by \cref{eq:phik} is then the entrywise
extension of $\phi_k$ from $\mathcal{R}_{2^k}$ to
$\matrices(\mathcal{R}_{2^k})$. We can therefore use \cref{cor:detrel}
to get, for $U\in\unitaries(\mathcal{R}_{2^k})$, an expression for the
determinant of $\phi_k(U)$.

\begin{corollary}
\label{cor:detblock16}
    Let $k \geq 2$ and let $U \in \unitaries_n(\mathcal{R}_{2^k})$. Then we have 
    \[
    \det\nolimits_{\mathcal{R}_{2^{k-1}}}(\phi_{k}(U)) = \det\nolimits_{\mathcal{R}_{2^{k-1}}}(\phi_k(\det\nolimits_{\mathcal{R}_{2^{k}}}(U))).
    \]
\end{corollary}

In other words, \cref{cor:detblock16} states that, for
$U\in\unitaries(\mathcal{R}_{2^k})$, to compute the determinant of
$\phi_k(U)$ over $\mathcal{R}_{2^{k-1}}$, one can first compute the
determinant $u$ of $U$ over $\mathcal{R}_{2^k}$, and then compute the
determinant of $\phi_k(u)$ over $\mathcal{R}_{2^{k-1}}$.

\begin{remark}
The function $\det\nolimits_{\mathcal{R}_{2^{k-1}}}\circ \phi_k :
\mathcal{R}_{2^k} \to \mathcal{R}_{2^{k-1}}$ is known in number theory
as the \textbf{relative norm} of the field extension $\Q(\zeta_{2^k})
/ \Q(\zeta_{2^{k-1}})$, and is often denoted by
$\mathrm{N}_{\Q(\zeta_{2^k}) /
  \Q(\zeta_{2^{k-1}})}$. \cref{cor:detblock16} therefore states that
the determinant (over $\mathcal{R}_{2^{k-1}}$) of $\phi_k(U)$ is the
relative norm of the determinant (over $\mathcal{R}_{2^{k}}$) of $U$,
i.e., $\det\nolimits_{\mathcal{R}_{2^{k-1}}} \circ \phi_k =
\mathrm{N}_{\Q(\zeta_{2^k}) / \Q(\zeta_{2^{k-1}})} \circ
\det\nolimits_{\mathcal{R}_{2^{k}}}$.
\end{remark}

We close this section by introducing a second catalytic embedding. The
construction is essentially the same as in the definition of
$(\phi_k,\textbf{c}_k)$. We define the matrix $\Gamma_k$ and the
vector $\textbf{d}_k$ by
    \[
    \Gamma_{k} = \begin{bmatrix}
        0 & 1 \\
        \zeta_{3\cdot 2^{k-1}} & 0
    \end{bmatrix} 
    \qquad
    \mbox{and}
    \qquad
    \textbf{d}_k = \frac{1}{\sqrt{2}}\begin{bmatrix}
        1\\
        \zeta_{3\cdot 2^k}
    \end{bmatrix}.
    \]
We then define the function $\psi_{k}: \matrices(\mathcal{R}_{3\cdot
  2^k}) \to \matrices(\mathcal{R}_{3\cdot 2^{k-1}})$ by
    \[
    \psi_k(A+B\zeta_{2^k}) = A \otimes I_2 + B \otimes \Gamma_{k}.
    \]
Reasoning as above, we obtain the proposition below.

\begin{proposition}
    \label{prop:cyclocatemb2}
    Let $k \geq 2$. Then $(\psi_{k}, \textbf{d}_k)$ is a
    $2$-dimensional catalytic embedding from
    $\unitaries(\mathcal{R}_{3\cdot 2^k})$ to
    $\unitaries(\mathcal{R}_{3\cdot 2^{k-1}})$.
\end{proposition}

An analogue of \cref{cor:detblock16} holds for
$(\psi_k,\textbf{d}_k)$, but we omit it here, since it will not be
useful for our purposes.

% --------------------------------------------------------------------
\section{Clifford-cyclotomic circuits of degree \texorpdfstring{$2^k$}{2ᵏ}}
\label{sec:16}

We now turn to the exact synthesis of circuits for matrices with
entries in the ring $\mathcal{R}_{2^k}$. We will take advantage of the
following result, which was established in \cite[Lemma~7]{Giles2013a}.

\begin{theorem}[Giles \& Selinger]
\label{thm:gsone}
If $U$ is a $2^m \times 2^m$ matrix with entries in $\mathcal{R}_8$
and $\det(U)=1$, then $U$ can be exactly represented by an
ancilla-free $m$-qubit circuit over $\mathcal{G}_8$.
\end{theorem}

Let $U$ be a unitary matrix with entries in $\mathcal{R}_{16}$. Then
the determinant of $U$ over $\mathcal{R}_{16}$ is an element of
$\mathcal{R}_{16}$ of norm 1. The next lemma shows that this
determinant must be a power of $\z_{16}$. The proof is relegated to
\cref{app:lengthone}.

\begin{lemma}
    \label{lem:dyadicunit}
    Let $u \in \mathcal{R}_{16}$ be such that $\abs{u} = 1$. Then $u =
    \zeta_{16}^\ell$ for some $0 \leq \ell \leq 15$.
\end{lemma}

In order to save an ancilla in synthesizing circuits over
$\mathcal{G}_{16}$ our strategy is the following. Consider an
$m$-qubit unitary $U$ with entries in $\mathcal{R}_{16}$. By
\cref{lem:dyadicunit}, the determinant of $U$ is a power of
$\zeta_{16}$. Hence, multiplying $U$ by an appropriate power of the
one-level operator of type $\z_{16}$ if needed, we can assume without
loss of generality that $U$ has determinant 1. By
\cref{cor:detblock16}, $\phi_4(U)$ is then an $(m+1)$-qubit unitary of
determinant 1 with entries in $\mathcal{R}_8$ and it can thus be
represented by an $(m+1)$-qubit circuit by \cref{thm:gsone}.

\begin{theorem}
\label{thm:sixteenth}
    A $2^m \times 2^m$ matrix $U$ can be exactly represented by an
    $m$-qubit circuit over $\mathcal{G}_{16}$ if and only if $U \in
    \unitaries_{2^m}(\mathcal{R}_{16})$. Furthermore, a single ancilla
    suffices to synthesize a circuit for $U$.
\end{theorem}

\begin{proof}
    The left-to-right implication follows immediately from the fact
    that all the elements of $\mathcal{G}_{16}$ belong to
    $\unitaries(\mathcal{R}_{16})$. For the right-to-left implication,
    let $U \in \unitaries_{2^m}(\mathcal{R}_{16})$. Since $U$ is
    unitary, $\det\nolimits_{\mathcal{R}_{16}}(U)$ is an element of
    $\mathcal{R}_{16}$ of norm $1$. Thus, by \cref{lem:dyadicunit}, we
    have that $\det\nolimits_{\mathcal{R}_{16}}(U) = \zeta_{16}^\ell$
    for some $0 \leq \ell \leq 15$. Now let $P$ be the
    fully-controlled phase gate
    \[
    P = \diag(1, 1, ..., 1, \zeta_{16}^\ell),
    \]
    and let $V = P^\dagger U$. Note that
    $\det\nolimits_{\mathcal{R}_{16}}(V) =
    \det\nolimits_{\mathcal{R}_{16}}(P^\dagger)
    \det\nolimits_{\mathcal{R}_{16}}(U) = 1$. Now consider the
    catalytic embedding $(\phi_{4}, \textbf{c}_4)$ defined in
    \cref{sec:catemb}. By \cref{cor:detblock16}, we have
    \[
    \det\nolimits_{\mathcal{R}_{8}}(\phi_{4}(V)) = \det\nolimits_{\mathcal{R}_{8}}(\phi_{4}(\det\nolimits_{\mathcal{R}_{16}}(V)))= \det\nolimits_{\mathcal{R}_{8}} (I_2) = 1.
    \]
    Hence, $\phi_{4}(V) \in \unitaries_{2^{m+1}}(\mathcal{R}_{8})$ and
    $\det\nolimits_{\mathcal{R}_{8}}(\phi_{4}(V))=1$ so that, by
    \cref{thm:gsone}, $\phi_{4}(V)$ can be represented by an
    ancilla-free circuit $C$ over $\mathcal{G}_8$. Now let $D$ be the
    circuit $D = (I_m \otimes (T_{16} H))^\dagger \circ C \circ (I_m
    \otimes (T_{16} H))$ over $\mathcal{G}_{16}$.  Then
    $T_{16}H\textbf{e}_0 = \textbf{c}_4$, so that, for an arbitrary
    $\textbf{u}$, we have:
    \begin{align*}
        D(\textbf{u}\otimes\textbf{e}_0) &= (I_m \otimes (T_{16} H))^\dagger \circ C \circ (I_m \otimes (T_{16} H)) (\textbf{u}\otimes\textbf{e}_0) \\
        &= (I_m \otimes (T_{16} H))^\dagger \circ C(\textbf{u}\otimes\textbf{c}_4) \\
        &= (I_m \otimes (T_{16} H))^\dagger \circ \phi_{4}(V)(\textbf{u}\otimes\textbf{c}_4) \\
        &= (I_m \otimes (T_{16} H))^\dagger ((V\textbf{u})\otimes\textbf{c}_4) \\
        &= (V\textbf{u})\otimes\textbf{e}_0.
    \end{align*}
    Hence $D$ exactly represents $V$ over $\mathcal{G}_{16}$ using a
    single ancilla. Moreover, $P$ can also be represented by a circuit
    over $\mathcal{G}_{16}$ using a single ancilla. Indeed, this
    follows from \cref{thm:gstwo}, since $P$ is (a power of) a
    one-level operator of type $\zeta_{16}$. Thus, the unitary $U$ can
    be represented using a single ancilla as well.
\end{proof}

The circuit constructed in the proof of \cref{thm:sixteenth} is
depicted in \cref{fig:sixteenth}. We note that, if $U$ is a unitary
with entries in $\mathcal{R}_{16}$ and the dimension of $U$ is larger
than 4, then an ancilla is necessary to synthesize a matrix for a
$U$. As a result, the construction of \cref{thm:sixteenth} uses the
minimal number of ancillas (in the worst case).

\begin{figure}[t]
\centering
\begin{quantikz}
\lstick{$\textbf{u}$} 
& \gate[wires=2]{~~P~~} \gategroup[2,steps=1,style={dashed,rounded corners,fill=red!20, inner xsep=2pt},background,label style={label position=below,anchor=north,yshift=-0.2cm}]{{}}  
& \push{~P\textbf{u}~}
& \qw \gategroup[2,steps=5,style={dashed,rounded corners,fill=blue!20, inner xsep=2pt},background,label style={label position=below,anchor=north,yshift=-0.2cm}]{{}} 
& \qw 
& \gate[wires=2]{~~C~~} 
& \qw 
& \qw 
& \qw\rstick{$U\textbf{u}$}\\
\lstick{$\textbf{e}_{0}$}  
&
& \push{~\textbf{e}_{0}~}
& \gate{H}\vphantom{T_{16}^\dagger}\hphantom{~} 
& \gate{T_{16}}\vphantom{T_{16}^\dagger} 
& 
& \gate{T_{16}^\dagger}\vphantom{T_{16}^\dagger} 
& \gate{H}\vphantom{T_{16}^\dagger}\hphantom{~} 
& \qw\rstick{$\textbf{e}_{0}$}      
\end{quantikz}   
\caption{The circuit constructed in the proof of
  \cref{thm:sixteenth}.\label{fig:sixteenth}}
\end{figure}
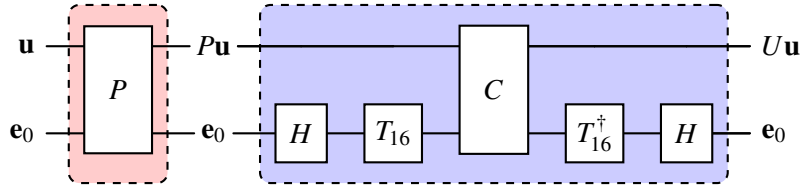

\begin{corollary}
\label{cor:16catemb}
    Let $k \geq 4$. A $2^m \times 2^m$ matrix $U$ can be exactly
    represented by an $m$-qubit circuit over $\mathcal{G}_{2^k}$ if
    and only if $U \in U_{2^m}(\mathcal{R}_{2^k})$. Furthermore, $k-3$
    ancillas suffice to synthesize a circuit for $U$.
\end{corollary}

\begin{proof}
    By induction, using \cref{thm:sixteenth} as the base case, and a
    construction similar to the one given in \cref{fig:sixteenth} in
    the inductive step.
\end{proof}

% --------------------------------------------------------------------
\section{Clifford-cyclotomic circuits of degree \texorpdfstring{$3\cdot 2^k$}{3*2ᵏ}}
\label{sec:12}

We now turn to the exact synthesis of circuits for matrices with
entries in the ring $\mathcal{R}_{3\cdot 2^k}$. Our strategy is to
first prove that every matrix with entries in $\mathcal{R}_{24}$ can
be represented by a circuit over $\mathcal{G}_{24}$, and to then use
an inductive argument like the one used in proving
\cref{cor:16catemb}. In order to establish the result for $n=24$, we
start by showing that every unitary matrix with entries in
$\mathcal{R}_{12}$ can be factored as a product of convenient generators.

% --------------------------------------------------------------------
\subsection{Generating \texorpdfstring{$\unitaries_n(\Dzt)$}{Un(R[z12]}}
\label{ssec:gen12}

To generate $\unitaries_n(\Dzt)$, we use one-level operators
of type $\zt$ and two-level operators of type $X$ and $H'$, where 
\[
H' = \zeta_8 H = 
\frac{\delta}{2} 
\begin{bmatrix} 
  1 & 1 \\ 
  1 & -1
\end{bmatrix}.
\]
Note that $H'\in \unitaries(\Dzt)$. We follow
\cite{Amy2020,Giles2013a} and first show that a unit vector with
entries in $\Dzt$ can be mapped to a standard basis vector using one-
and two-level operators of type $\zeta_{12}$, $X$, and $H'$. To this
end, we proceed by induction on the least denominator exponent of the
vector.

\begin{lemma}
    \label{lem:12thbaseineq}
    Let $a$ and $b$ be integers, at least one of which is
    nonzero. Then $a^2 + b^2 + ab \geq 1$, and equality is achieved
    exactly when $(a,b)$ is one of $(\pm 1, 0)$, $(0, \pm 1)$, or
    $(\pm 1, \mp 1)$.
\end{lemma}

\begin{proof}
    First notice that if $a\neq 0$ and $b = 0$, then $a^2 + b^2 + ab =
    a^2 \geq 1$, since $a \in \Z$. A similar argument applies when $a
    = 0$ and $b \neq 0$. Equality in the first of these cases happens
    exactly when $a=\pm 1$ and $b=0$, while equality in the second of
    these cases happens exactly when $a=0$ and $b=\pm 1$. Now suppose
    that $a$ and $b$ are both nonzero, and assume, without loss of
    generality, that $\abs{a} \geq \abs{b}$. We then have $a^2 \geq
    \abs{ab}$, so that $a^2 +ab \geq 0$. Hence,
    \[ 
    a^2 + ab + b^2 \geq b^2 \geq 1,
    \]
    since $b\in\Z$. Equality in this case happens exactly when $b^2
    =1$ and $a^2 + ab = 0$ which, in turn, happens exactly when $b=\pm
    1$ and $a=\mp 1$.
\end{proof}

\begin{lemma}
\label{lem:12thbase}
    If $\textbf{u}$ is an $m$-dimensional unit vector with entries in
    $\mathcal{R}_{12}$ and $\lde(\textbf{u})=0$, then, for any $0\leq
    j \leq m-1$, there exists a sequence $G_1, \ldots, G_q$ of one-
    and two-level operators of type $\zeta_{12}$, $X$, and $H'$ such
    that $G_1 \cdots G_q \textbf{u}= \textbf{e}_j$.
\end{lemma}

\begin{proof}
    It suffices to show that $\textbf{u}= \z_{12}^\ell\textbf{e}_{j'}$
    for some integers $j'$ and $\ell$, since, if $\textbf{u}$ is of
    that form, then it can be mapped to $\textbf{e}_j$ by applying the
    appropriate operators of type $\z_{12}$ and $X$. Because
    $\textbf{u}$ is a unit vector, we have $\textbf{u}^\dagger
    \textbf{u} = 1$. Hence, using \cref{eq:norm12}, we get
    \[
    1 = \textbf{u}^\dagger \textbf{u} 
    = \sum_{j} u_j^\dagger u_j
    = \sum_{j} ((a_j^2 + c_j^2 + a_jc_j) + (b_j^2 + d_j^2 + b_jd_j))+(a_j b_j + b_j c_j + c_j d_j)\sqrt{3}.
    \]
    Since every element $\Z[\sqrt{3}]$ can be uniquely expressed as an
    integer linear combination of 1 and $\sqrt{3}$, the equation above
    implies the equations below.
    \begin{align}
        &\sum_{j} ((a_j^2 + c_j^2 + a_jc_j) + (b_j^2 + d_j^2 + b_jd_j)) = 1 \label{cond1} \\
        &\sum_{j} (a_j b_j + b_j c_j + c_j d_j) = 0 \label{cond2}
    \end{align}
    It follows from \cref{lem:12thbaseineq} that \cref{cond1} can only
    be satisfied if there is exactly one index $j$ such that either
    $a_j^2 + c_j^2 +a_jc_j=1$ or $b_j^2 + d_j^2 +b_jd_j =1$, but not
    both. By \cref{lem:12thbaseineq}, this happens precisely when
    $(a_j,c_j)$ is one of $(\pm 1, 0)$, $(0, \pm 1)$, or $(\pm 1, \mp
    1)$, and $(b_j,d_j)$ is $(0,0)$, or when $(b_j,d_j)$ is one of
    $(\pm 1, 0)$, $(0, \pm 1)$, or $(\pm 1, \mp 1)$, and $(a_j,c_j)$
    is $(0,0)$. These 12 solutions all satisfy \cref{cond2} and
    correspond exactly to the possible powers of $\zt$.
\end{proof}

\begin{lemma}
\label{lem:hreduce}
If $u,v\in\Zzt$ are such that $u^\dagger u \equiv_2 v^\dagger v$, then
there exists $\ell$ such that
\[
H'T_{12}^\ell \begin{bmatrix} u \\ v \end{bmatrix} = \begin{bmatrix} u' \\ v' \end{bmatrix}
\]
for some $u',v'\in\Zzt$ such that $u' \equiv v' \equiv 0
\pmod{\delta}$.
\end{lemma}

\begin{proof}
Let $u$ and $v$ be as stated. Then, by \cref{lem:res}, we have
$u\equiv_2 \zeta_{12}^\ell v$ for some $\ell$. Equivalently, $u \pm \zt^\ell v
\equiv 0 \pmod{2}$, so that $u + \zt^\ell v = 2 w$ and $u - \zt^\ell
v = 2w'$ for some $w,w' \in \Zzt$. We then get
    \[
        H'T_{12}^\ell \begin{bmatrix} u \\ v \end{bmatrix} = H'\begin{bmatrix}
            u \\
            \zt^\ell v 
        \end{bmatrix}
        = \frac{\delta}{2}\begin{bmatrix}
            u + \zt^\ell v\\
            u - \zt^\ell v 
        \end{bmatrix}
        = \frac{\dzt}{2}\begin{bmatrix}
            2w\\
            2w'
        \end{bmatrix} = \dzt\begin{bmatrix}
            w\\
            w'
        \end{bmatrix},
    \]
    which completes the proof.
\end{proof}

\begin{lemma}
    \label{lem:12thstep}
    If $\textbf{u}$ is an $m$-dimensional unit vector with entries in
    $\mathcal{R}_{12}$ and $\lde(\textbf{u}) \geq 1$, then there
    exists a sequence $G_1, \ldots G_q$ of one- and two-level
    operators of type $\z_{12}$, $X$, and $H'$ such that $\lde(G_1
    \cdots G_q \textbf{u}) < \lde(\textbf{u})$.
\end{lemma}

\begin{proof}
    Let $k=\lde (\textbf{u})$ and let $\textbf{v} = \dzt^k\textbf{u}
    \in \Zzt$. We know from \cref{lem:res} that, for $v\in \Zzt$, if
    $v\not\equiv_{\delta} 0$, then $v^\dagger v\equiv_2 1$ or
    $v^\dagger v\equiv_2 \sqrt{3}$. We can therefore write
    $\textbf{v}^\dagger \textbf{v}$ as
    \begin{equation}
    \label{eq:synthstep}
    \textbf{v}^\dagger \textbf{v} 
    = \sum_j v_j^\dagger v_j 
    = \sum_{v_j\equiv_{\delta} 0} v_j^\dagger v_j + \sum_{v_j\not\equiv_{\delta} 0} v_j^\dagger v_j
    = \sum_{v_j\equiv_{\delta} 0} v_j^\dagger v_j + \sum_{v_j^\dagger v_j \equiv_2 1} v_j^\dagger v_j   + \sum_{v_j^\dagger v_j \equiv_2 \sqrt{3}} v_j^\dagger v_j.    
    \end{equation}
    Since $\textbf{u}$ is a unit vector and $\delta^\dagger \delta
    =2$, we have $\textbf{v}^\dagger\textbf{v} =
    \textbf{u}^\dagger\textbf{u}(\dzt^\dagger\dzt)^k = 2^k$. Because
    $k\geq 1$, this implies that $\textbf{v}^\dagger \textbf{v}
    \equiv_2 0$. Hence, by \cref{lem:res}, taking \cref{eq:synthstep}
    modulo 2 yields
    \[
    0 \equiv_2 \textbf{v}^\dagger \textbf{v} \equiv_2  a + b\sqrt{3},
    \]
    where $a$ is the number of $v_j$ such that $v_j^\dagger v_j
    \equiv_2 1$, and $b$ be the number of $v_j$ such that $v_j^\dagger
    v_j \equiv_2 \sqrt{3}$. It then follows that we must have
    $a\equiv_2 b \equiv_2 0$. That is, both $a$ and $b$ are even
    integers. We can therefore group the entries of $\textbf{v}$ that
    are not congruent to 0 modulo $\delta$ into pairs $(v_j, v_{j'})$
    such that $v_j^\dagger v_j \equiv_2 v_{j'}^\dagger
    v_{j'}$. Applying \cref{lem:hreduce} to every such pair reduces
    the least denominator exponent of $\textbf{u}$.
\end{proof}

\begin{lemma}
    \label{lem:column}
    If $\textbf{u}$ is an $m$-dimensional unit vector with entries in
    $\mathcal{R}_{12}$, then, for any $0\leq j \leq m-1$, there exists
    a sequence $G_1, \ldots, G_q$ of one- and two-level operators
    of type $\zeta_{12}$, $X$, and $H'$ such that $G_1 \cdots G_q
    \textbf{u}= \textbf{e}_j$.
\end{lemma}

\begin{proof}
By induction on $\lde(\textbf{u})$. If $\lde(\textbf{u})=0$, then the
result follows from \cref{lem:12thbase}. If $\lde(\textbf{u})\geq 1$,
then, by \cref{lem:12thstep}, there exists a sequence $G_1, \ldots
G_q$ of two-level operators of type $\zeta_{12}$, $X$, and $H'$ such
that $\lde(G_1 \cdots G_q \textbf{u}) < \lde(\textbf{u})$. Now let
$\textbf{u}' = G_1 \cdots G_q \textbf{u}$. By the induction
hypothesis, there exists a sequence $G_1', \ldots, G_{q'}'$ of one-
and two-level operators of type $\zeta_{12}$, $X$, and $H$ such that
$G_1' \cdots G_{\ell'}' \textbf{u}'= \textbf{e}_j$. We therefore have
\[
\textbf{e}_j = G_1' \cdots G_{q'}' \textbf{u}' = G_1' \cdots
G_{q'}' \cdot G_1 \cdots G_q \textbf{u},
\]
which completes the proof.
\end{proof}

\begin{theorem}
    \label{thm:12th}
    A matrix $U$ belongs to $\unitaries(\mathcal{R}_{12})$ if and only
    if $U$ can be expressed as a product of one- and two-level
    operators of type $\z_{12}$, $X$, and $H'$.
\end{theorem}

\begin{proof}
    The right-to-left direction follows immediately from the fact that
    one- and two-level operators of type $\zeta_{12}$, $X$, and $H'$
    are unitaries with entries in $\mathcal{R}_{12}$. We now prove the
    left-to-right direction. Let $U \in \unitaries(\mathcal{R}_{12})$
    and let $\textbf{u}$ be the first column of $U$. Then $\textbf{u}$
    is a unit vector with entries in $\mathcal{R}_{12}$. Hence, by
    \cref{lem:column}, there exists a sequence $G_1, \ldots, G_q$
    of one- and two-level operators of type $\zeta_{12}$, $X$, and $H'$
    such that $G_1 \cdots G_q \textbf{u}= \textbf{e}_1$.  Since $U$
    and $G_1, \ldots, G_q$ are unitaries with entries in
    $\mathcal{R}_{12}$, we have
    \[
    G_1\cdots G_q U =
    \left[\begin{array}{c|ccc}
    1&0&\cdots&0 \\ \hline
    0&&& \\
    \vdots&&\text{$U'$}& \\
    0&&&
    \end{array}\right],
    \] 
    for some smaller $U' \in \unitaries(\mathcal{R}_{12})$. Repeating
    this process inductively, we ultimately obtain a sequence $F_1,
    \ldots, F_{q'}$ of one- and two-level operators of type $\z_{12}$,
    $X$ and $H'$ such that $F_1 \cdots F_{q'}U = I$. Multiplying by
    $(F_1 \cdots F_{q'})^{-1}$ on both sides then yields the desired
    decomposition.
\end{proof}

% --------------------------------------------------------------------
\subsection{Exact synthesis}
\label{ssec:exactsynth12}

\begin{theorem}
\label{thm:24th}
    A $2^m \times 2^m$ matrix $U$ can be exactly represented by an
    $m$-qubit circuit over $\mathcal{G}_{24}$ if and only if $U \in
    \unitaries_{2^m}(\mathcal{R}_{24})$. Furthermore, 2 ancillas
    suffice to synthesize a circuit for $U$.
\end{theorem}

\begin{proof}
    The left-to-right follows immediately from the fact that elements
    of $\mathcal{G}_{24}$ belong to
    $\unitaries(\mathcal{R}_{24})$. Now let $U \in
    \unitaries_{2^m}(\mathcal{R}_{24})$ and let
    $(\psi_{3},\textbf{d}_3): \unitaries(\mathcal{R}_{24}) \to
    \unitaries(\mathcal{R}_{12})$ be the catalytic embedding defined
    in \cref{sec:catemb}. Then $\psi_{3}(U) \in
    \unitaries_{2^{m+1}}(\mathcal{R}_{12})$ and can be represented as
    a product of one- and two-level operators of type $\zt$, $X$, and
    $H'$. By \cref{thm:gstwo}, these 1- and 2-level operators can each
    be represented by a circuit over $\mathcal{G}_{24}$ using a single
    ancilla. Hence, there is a circuit $C$ over $\mathcal{G}_{24}$
    that represents $\psi_3(U)$. Using an additional ancilla and
    reasoning as in \cref{thm:sixteenth}, we obtain a circuit over
    $\mathcal{G}_{24}$ for $U$ that uses 2 ancillas.
\end{proof}

We can now use \cref{thm:24th} to obtain an exact synthesis result for
Clifford-cyclotomic gate sets of degree $3\cdot 2^k$, with $k\geq
3$. The proof is very similar to that of \cref{cor:16catemb}, so we
omit it here.

\begin{corollary}
    \label{cor:24catemb}
    Let $k\geq 3$. A $2^m \times 2^m$ matrix $U$ can be exactly
    represented by an $m$-qubit circuit over $\mathcal{G}_{3\cdot2^k}$
    if and only if $U \in
    U_{2^m}(\mathcal{R}_{3\cdot2^k})$. Furthermore, $k-1$ ancillas
    suffice to synthesize a circuit for $U$.
\end{corollary}

% --------------------------------------------------------------------
\section{Conclusion}
\label{sec:conc}

We now know that, for $n=2^k$ and $n=3\cdot 2^k$, $m$-qubit circuits
with ancillas over $\mathcal{G}_n$ correspond precisely to matrices in
$\unitaries_{2^m}(\mathcal{R}_n)$. Yet, many questions in the theory
of Clifford-cyclotomic circuits remain unanswered. Two natural open
problems are the following.
\begin{enumerate}
\item For which values of $n$ does the correspondence between circuits
  over $\mathcal{G}_n$ and matrices in $\unitaries(\mathcal{R}_n)$
  hold?
\item For the values of $n$ for which the correspondence holds, what
  is the smallest number of ancillas required to synthesize circuits
  in the worst case?
\end{enumerate}
A natural initial step in addressing the first of these open problems
is to consider values of $n$ of the form $p\cdot 2^k$, for $p$ a
prime. While it stands to reason that our results might generalize to
such cases, it gets progressively harder to analyze the relevant
residues; this indicates that a novel approach may be needed for such
generalizations.

% --------------------------------------------------------------------
\section*{Acknowledgements}

We thank the QPL reviewers for their very thoughtful comments. This
work was supported by the Natural Sciences and Engineering Research
Council of Canada (NSERC). The circuit in \cref{sec:16} was drawn
using the Quantikz package \cite{quantikz}.

% --------------------------------------------------------------------
\bibliographystyle{eptcs}
\bibliography{contributions}

\newpage

% --------------------------------------------------------------------
\appendix

% --------------------------------------------------------------------
\section{A Proof of \texorpdfstring{\cref{lem:dyadicunit}}{Lemma~1.2}}
\label{app:lengthone}

Let $R$ be a ring. A \textbf{unit} of $R$ is an element of $R$ that
admits a multiplicative inverse. A ring $R$ is an \textbf{integral
  domain} if $uv \neq 0$, for all nonzero elements $u,v \in R$.  An
element $u \in R$ is an \textbf{associate} of an element $v \in R$ if
there exists a unit $w \in R$ such that $v=wu$. A nonunit, nonzero
element $u \in R$ is called \textbf{prime} if $u$ divides $v$ or $u$
divides $w$, whenever $u$ divides $vw$. An ideal $I \varsubsetneq R$
is called \textbf{prime} if $uv \in I$ implies $u \in I$ or $v \in I$.

We start by recalling two well known facts, before establishing a
property of roots of unity.

\begin{proposition}
\label{prop:assocdiv}
    Let $R$ be a commutative ring with identity and let $u,v \in
    R$. Then $u$ and $v$ are associates if and only if $u \mathrel{|}
    v$ and $v \mathrel{|} u$.
\end{proposition}

\begin{proposition}
    \label{lem:primeideal}
    Let $R$ be a commutative ring with identity and let $u \in
    R$. Then $u$ is prime if and only if the ideal generated by $u$ is
    prime.
\end{proposition}

\begin{lemma}
    \label{lem:deltaassoc}
    Let $a, b \in \N$ be such that $\gcd(a,b) = 1$. Then $1-\zeta_a^b$
    and $1-\zeta_a$ are associates.
\end{lemma}

\begin{proof}
    We have
    \[
    1-\zeta_a^b = (1-\zeta_a)(1+\zeta_a+\zeta_a^2+...+\zeta_a^{b-1}).
    \] 
    Hence, $1-\zeta_a \mathrel{|} 1-\zeta_a^b$.  Now since
    $\gcd(a,b)=1$, there exist $c,d \in \Z$ such that $ac + bd =
    1$. Then $1-\zeta_a = 1-\zeta_a^{ac+bd} = 1-(\zeta_a^b)^d$, and
    thus
    \[
    1-\zeta_a = 1-(\zeta_a^b)^d = (1-\zeta_a^b)(1+\zeta_a^b+(\zeta_a^b)^2+...+(\zeta_a^b)^{d-1}).
    \] 
    Hence, $1-\zeta_a^b \mathrel{|} 1-\zeta_a$. Thus, $1-\zeta_a$ and
    $1-\zeta_a^b$ are associates by \cref{prop:assocdiv}.
\end{proof}

We are now in a position to prove \cref{lem:dyadicunit}, whose
statement we reproduce below.

\begin{lemma*}
    Let $u \in \mathcal{R}_{16}$ be such that $\abs{u} = 1$. Then $u =
    \zeta_{16}^\ell$ for some $0 \leq \ell \leq 15$.
\end{lemma*}

\begin{proof}
    Let $\chi = 1-\zeta_{16} \in \mathcal{R}_{16}$. Then $\chi$ is a
    prime element in $\Z[\zeta_{16}]$ and, by \cref{lem:primeideal},
    $\langle \chi \rangle$ is a prime ideal of $\mathcal{R}_{16}$. By
    \cref{lem:deltaassoc}, we can decompose $2$ in $\Z[\zeta_{16}]$ as
    \begin{align*}
        2 &= (1+i)(1-i) \\
        &= (1-i)^2 u_1 \\
        &= (1-\zeta_8)^2(1+\zeta_8)^2 u_1 \\
        &= (1-\zeta_8)^4 u_2 u_1 \\
        &= (1-\zeta_{16})^4 (1+\zeta_{16})^4 u_2 u_1 \\
        &= (1-\zeta_{16})^8 u_3 u_2 u_1 \\
        &= \chi^8 u_3 u_2 u_1,
    \end{align*}
    where $u_1$, $u_2$, and $u_3$ are units in $\Z[i]$, $\Z[\zeta_8]$,
    and $\Z[\zeta_{16}]$, respectively. Hence, any $u \in
    \mathcal{R}_{16}$, can be written as $u = v/\chi^{\ell}$ with
    $\ell \in \N$ and $v\in \Z[\zeta_{16}]$. Now let $u\in
    \mathcal{R}_{16}$ be such that $\abs{u}=1$ and write $u$ as $u =
    v/\chi^\ell$ with $\ell$ minimal.  Observe that $\chi^\dagger = 1
    - \zeta_{16}^\dagger = -\zeta_{16}^\dagger\chi$. We hence have
    \[
    1 = \abs{u}^2 = u^\dagger u = \frac{v^\dagger v}{(-\zeta_{16}^{\dagger})^\ell \chi^{2\ell}},
    \]
    so that $v^\dagger v = \chi^{2\ell}\cdot
    (-\zeta^{\dagger}_{16})^\ell$. If $\ell > 0$, we have $v^\dagger v
    \in \langle \chi \rangle$. Since $\chi$ is prime, this implies
    that we must have either $v \in \langle \chi \rangle$ or
    $v^\dagger \in \langle \chi \rangle$. But $v^\dagger \in \langle
    \chi \rangle$ would imply $v \in \langle \chi \rangle$ since
    $\chi$ and $\chi^\dagger$ are associates. Hence, if $\ell>0$, then
    $v\in \langle \chi \rangle$, which contradicts the minimality of
    $\ell$. It must thus be the case that $\ell=0$, so that $u$ is in
    fact an element of $\Z[\zeta_{16}]$. Writing $u$ as $u =
    \sum_{j=0}^7 a_j\zeta_{16}^j$, with $a_j \in \Z$, we then get
    \begin{align*}
        1 &= u^\dagger u \\
        &= \sum_{j=0}^7 a_j^2 + 2Re(\zeta_{16})((\sum_{j=0}^6 a_j a_{j+1}) - a_0 a_7) + 2Re(\zeta_{16}^2)((\sum_{j=0}^5 a_j a_{j+2}) - a_0 a_6 - a_1 a_7) \\
        &+ 2Re(\zeta_{16}^3)((\sum_{j=0}^4 a_j a_{j+3}) - a_0 a_5 - a_1 a_6 - a_2 a_7).
    \end{align*}
    The above equation holds when exactly one $a_j = \pm 1$ and the
    rest are zero. Hence $u = \zeta_{16}^\ell$ for some $0 \leq \ell
    \leq 15$, as desired.
\end{proof}

% --------------------------------------------------------------------
\end{document}